\documentclass[aps,prl,letterpaper,reprint,twocolumn,floatfix,superscriptaddress,notitlepage,usenames,dvipsnames,svgnames,x11names,table,nofootinbib,longbibliography]{revtex4-2}

\pdfoutput=1
\usepackage{graphicx, color, graphpap}      % Include figure files
\usepackage{amsmath}
\usepackage{enumerate}
\usepackage{amssymb}
\usepackage{amsthm}
\usepackage{pstricks}
\usepackage{float}

\usepackage[pdfusetitle]{hyperref}
\hypersetup{pdflang={English},colorlinks=true,linkcolor=RoyalBlue,citecolor=RoyalBlue,urlcolor=ForestGreen,breaklinks=true}

\usepackage[T1]{fontenc}
\usepackage{bbm}
\usepackage{dsfont}
\usepackage[linesnumbered,ruled,vlined]{algorithm2e}
\SetKwInput{kwInit}{Init}
\usepackage{mathtools}

\usepackage{tikz}
\usetikzlibrary{positioning}
\usepackage{xcolor}
\usepackage[caption=false]{subfig}
\usepackage[ruled,vlined]{algorithm2e}
\usepackage{siunitx}
\usepackage{qcircuit}
\usepackage{pifont} % http://ctan.org/pkg/pifont
\usepackage{physics}
\usepackage{thmtools}
\usepackage{thm-restate}


\renewcommand{\eqref}[1]{(\ref{#1})}
\newtheoremstyle{example}{\topsep}{\topsep}%
{}%         Body font
{}%         Indent amount (empty = no indent, \parindent = para indent)
{\bfseries}% Thm head font
{:}%        Punctuation after thm head
{   }%     Space after thm head (\newline = linebreak)
{\thmname{#1}\thmnumber{ #2}}%\thmnote{ #3}}%         Thm head spec
\theoremstyle{example}
\newtheorem{theorem}{Theorem}
\newtheorem{lemma}{Lemma}

\theoremstyle{definition}

\newtheorem*{theorem*}{Theorem}

\def\orcid#1{\kern -0.4em\href{https://orcid.org/#1}{\includegraphics[keepaspectratio,width=0.7em]{orcid_logo.pdf}}}

\renewcommand{\H}{\mathcal{H}}

\usepackage{lipsum}

\long\def\ca#1\cb{} %Use for commenting out: \ca...\cb
\begin{document}
\title{Efficient classical simulation of large-scale unitary cluster Jastrow circuits}

\author{Hrishikesh Belagali}
\affiliation{Department of Computational Mathematics, Science, and Engineering, Michigan State University, East Lansing, MI 48823, USA}
\affiliation{Center for Quantum Computing, Science, and Engineering, Michigan State University, East Lansing, MI 48823, USA}

\author{Thomas Van Camp}
\affiliation{Department of Computer Science and Engineering, University of Michigan, Ann Arbor, MI 48109}
\affiliation{Center for Quantum Computing, Science, and Engineering, Michigan State University, East Lansing, MI 48823, USA}

\author{R. Pradeep}
\affiliation{Department of Physical Sciences, IISER Kolkata, Mohanpur, West Bengal 741246, India}

\author{Sourin Das}
\affiliation{Department of Physical Sciences, IISER Kolkata, Mohanpur, West Bengal 741246, India}

\author{Namit Anand}
\affiliation{HPE Quantum, Emergent Machine Intelligence, Hewlett Packard Labs, CA, USA}

\author{Ryan LaRose}
\thanks{Corresponding author:  \href{rmlarose@msu.edu}{rmlarose@msu.edu}}
\affiliation{Department of Computational Mathematics, Science, and Engineering, Michigan State University, East Lansing, MI 48823, USA}
\affiliation{Department of Electrical and Computer Engineering, Michigan State University, East Lansing, MI 48823, USA}
\affiliation{Department of Physics and Astronomy, Michigan State University, East Lansing, MI 48823, USA}
\affiliation{Center for Quantum Computing, Science, and Engineering, Michigan State University, East Lansing, MI 48823, USA}

\begin{abstract}
Recent experiments on quantum computers have challenged the limits of classical computation in chemistry, simulating ground states of strongly correlated molecules. Many of these experiments have utilized the unitary cluster Jastrow ansatz, a quantum circuit inspired by the unitary coupled cluster ansatz that can be tailored to current quantum hardware. Notably, the largest experiment in \href{https://doi.org/10.1126/sciadv.adu9991}{\textit{Sci. Adv.} \textbf{11}, 25 (2025)} executed a quantum circuit with 77 qubits and 10,570 gates on an IBM quantum computer and performed classical post-processing with up to 6400 nodes on Fugaku to compute ground state energies better than Hartree-Fock. In this work, we present a polynomial time classical algorithm to compute the energy of any single-layer unitary cluster Jastrow circuit, independent of locality constraints for quantum hardware. Our algorithm can reproduce the largest experiment from \href{https://doi.org/10.1126/sciadv.adu9991}{\textit{Sci. Adv.} \textbf{11}, 25 (2025)}  in less than a minute on a laptop, and through circuit optimization enabled by fast simulation we achieve a lower ground state energy than the experiment.%Further, we show that we can simulate single-layer unitary cluster Jastrow circuits for Hydrogen chains on up to 160 qubits, which is the largest number of qubits currently available on IBM quantum computers. 
\end{abstract}

% =============================================================================
% =============================================================================
\maketitle
% =============================================================================
% =============================================================================

% \tableofcontents

\textit{Introduction} ---
While recent experiments on quantum computers have demonstrated quantum advantage~\cite{Arute_Arya_Babbush_Bacon_Bardin_Barends_Biswas_Boixo_Brandao_Buell_,Zhong_Wang_Deng_Chen_Peng_Luo_Qin_Wu_Ding_Hu_2020,Wu_Bao_Cao_Chen_Chen_Chen_Chung_Deng_Du_Fan__2021,Zhong_Deng_Qin_Wang_Chen_Peng_Luo_Wu_Gong_Su_2021,Zhu_Cao_Chen_Chen_Chen_Chung_Deng_Du_Fan_Gong_2022,Madsen_Laudenbach_Askarani_Rortais_Vincent_Bulmer_Miatto_Neuhaus_Helt_Collins,Deng_Gu_Liu_Gong_Su_Zhang_Tang_Jia_Xu_Chen_2023,Kim_Eddins_Anand_Wei_van_Rosenblatt_Nayfeh_Wu_Zaletel_Temme_2023,dwave2024}, new and improved classical algorithms have challenged or refuted these claims~\cite{larose_brief_2024,Zhao_Zhong_Pan_Chen_Fu_Su_Xie_Zhao_Zhang_Ouyang__2024,Oh_Liu_Alexeev_Fefferman_Jiang_2024,Tindall_Fishman_Stoudenmire_Sels_2023,refute_ibm_1,refute_ibm_2,refute_ibm_3,refute_ibm_4}. Recently, there has been a  surge in large-scale quantum computations for chemistry following the development of sample-based quantum diagonalization (SQD)~\cite{Robledo2025} and the unitary cluster Jastrow (UCJ) ansatz~\cite{Matsuzawa_Kurashige_2020}, specifically the local UCJ (LUCJ)~\cite{motta_bridging_2023} tailored for current quantum hardware. These experiments include~\cite{Robledo2025} which executes LUCJ circuits with up to 77 qubits and 10,570 gates to compute the ground state energy of an iron sulfur cluster, as well as later experiments which execute LUCJ circuits with up to 94 qubits and 101,602 gates for molecular fragments of a protein-ligand complex~\cite{merz_crossing_2026}. Multiple other works have executed large-scale LUCJ circuits on quantum hardware~\cite{li_protein-ligand_2026,shirakawa_closed-loop_2025,shajan_molecular_2026}.

The (L)UCJ ansatz is appealing from both the physical perspective and from the hardware perspective. That is, the UCJ ansatz is related to the unitary coupled cluster singles and doubles (UCCSD) ansatz, for which it is known that the time to compute the energy is polynomial in the system size on a quantum computer and factorial on a classical computer~\cite{motta_bridging_2023}. And, from the hardware perspective, the LUCJ ansatz is appealing because it reduces the overhead of non-local interactions~\cite{necaise2026distributioncomplexityelectronicstructure}. The classical hardness of sampling from UCJ circuits has recently been studied, and it was shown that UCJ circuits can be used to perform arbitrary instantaneous quantum polynomial-time (IQP) computations, which are known to be classically hard to simulate assuming the polynomial hierarchy does not collapse~\cite{Hafid_Iwakiri_Tsubouchi_Yoshioka_Kohda_2025}. Also, classical algorithms based on matchgate + CPHASE simulation have recently been developed specifically for LUCJ circuits~\cite{Hassman_ReardonSmith_Ravi_Chong_Sung_2025}. Here, the simulation cost grows exponentially in the number of CPHASE gates, and the authors are able to simulate up to one layer of a 52 qubit LUCJ circuit. Other simulation strategies face similar challenges: non-local interactions in the UCJ ansatz are are challenging for tensor network simulation, and truncated Heisenberg evolution~\cite{begusic_real-time_2025} typically discards too many terms to achieve chemically relevant energies.

In this work, we provide a polynomial time classical algorithm for exact computation of the energy of any single-layer UCJ circuit, independent of locality. As a particular example, we simulate the 77 qubit, 10,570 gate experiment for the iron sulfur cluster from~\cite{Robledo2025}, which used a state-of-the-art superconducting qubit quantum computer and up to 6400 nodes on Fugaku, a top 10 supercomputer, in less than a minute on a laptop. Our algorithm computes energy (weak simulation) and does not sample bitstrings (strong simulation), and is therefore unable to run bitstring-based post-processing like SQD. Nonetheless, we show that circuit optimization enabled by fast energy calculation achieves a lower energy than achieved by SQD in~\cite{Robledo2025}. Additionally, we compute the energy of single-layer UCJ circuits for Hydrogen chains with up to 160 qubits, which is the largest number of qubits currently available on IBM and other state-of-the-art superconducting qubit quantum computers. Our algorithm is based on Heisenberg evolution of the fermionic Hamiltonian, as well as L\"{o}wdin's formula for computing the energy. Our work shows that, in the quest for quantum advantage in chemistry~\cite{Lee_Lee_Zhai_Tong_Dalzell_Kumar_Helms_Gray_Cui_Liu_2023}, single-layer UCJ circuits are insufficient for beyond-classical computation.

\textit{Unitary cluster Jastrow circuits} \label{sec:ucj-circuits} ---
The unitary cluster Jastrow (UCJ) ansatz was introduced in~\cite{Matsuzawa_Kurashige_2020} as a unitary variant of Neuscamman's cluster Jastrow operator. The local UCJ (LUCJ) ansatz was introduced in~\cite{motta_bridging_2023} to make the UCJ ansatz more amenable to current quantum hardware. The starting point is the unitary coupled cluster (UCC) ansatz~\cite{Anand_2022}
\begin{equation} \label{eqn:ucc-ansatz}
    |\psi\rangle_\text{UCC} := e^{T - T^\dagger} |\phi_0\rangle
\end{equation}where $T = \sum_i T_i $ is the skew-hermitian fermionic excitation operator and $|\phi_0\rangle$ is a reference state, canonically the Hartree-Fock state. While the sum over $i$ extends over all (single, double, triple, ...) fermionic excitations, a common cutoff is to include only single and double excitations in what is referred to as the UCC singles and doubles (UCCSD) ansatz. 
It is known that the time to compute the energy
\begin{equation}
    E = \langle \psi | \H | \psi \rangle_{\text{UCCSD}}
\end{equation}
for Hamiltonian $\H$ is polynomial in the system size on a quantum computer, and factorial in the system size on a classical computer~\cite{motta_bridging_2023}. Currently, the cost of implementing the UCC(SD) ansatz is prohibitive for current/near-future quantum computers.

\begin{figure}
    \centering
    \includegraphics[width=0.9\linewidth]{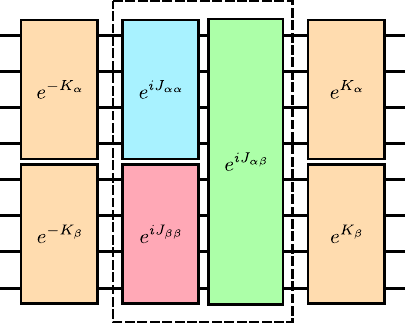}
    \caption{Structure of the unitary cluster Jastrow (UCJ) ansatz~\eqref{eqn:ucj-ansatz} for $L = 1$ layer. The top qubit lines are $\alpha$ spins and the bottom qubit lines are $\beta$ spins. The orange boxes can be implemented with $O(N^2)$ Givens rotations where $N = N_o + N_v$ is the sum of occupied and virtual orbitals. These operators do not change between the UCJ and local UCJ ansatz. In the local UCJ (LUCJ) ansatz, sparsity is imposed on the Jastrow $J = J_{\alpha \alpha} + J _{\beta \beta} + J_{\alpha \beta}$ operators (blue, red, and green operators in the dashed box) by removing interactions which are not supported on a given quantum computer.}
    \label{fig:ucj}
\end{figure}

The unitary cluster Jastrow (UCJ) ansatz can be related to a twice-factorized low-rank decomposition of the UCCSD ansatz. The UCJ ansatz is parameterized by a number of layers (or depth) $L$ and has the form
\begin{equation} \label{eqn:ucj-ansatz}
    |\psi\rangle_\text{UCJ} := \left[ \prod_{\mu = 1}^{L} e^{K_\mu} e^{i J_\mu} e^{-K_\mu} \right] |\phi_0\rangle 
\end{equation}
where $K_\mu$ and $J_\mu$ are one body operators of the form
\begin{equation} \label{eqn:orbital-rotation-K_mu}
    K_\mu := \sum_{p q, \sigma} k_{pq}^{\mu} a_{p \sigma}^\dagger a_{q \sigma}  % Note there is a typo in Eqn. (4) of https://pubs.rsc.org/sc/article/14/40/11213/827151/Bridging-physical-intuition-and-hardware?silentauthchecked=true which has a dagger on the a_{q \sigma}. Check Eqn. (6) of https://arxiv.org/abs/1909.12410 to verify.
\end{equation}
and
\begin{equation} \label{eqn:Jastrow-operator-J_mu}
    J_\mu := \sum_{pq, \sigma \tau} j_{pq, \sigma \tau}^\mu \hat{n}_{p \sigma} \hat{n}_{q \tau} .
\end{equation}
Here, $p, q \in [N]$ label molecular spatial orbitals, $\sigma, \tau$ label spin polarizations ($\alpha$ and $\beta$ for spin up and spin down electrons), $a^\dagger$ and $a$ are the fermionic creation and annihilation operators, $\hat{n}$ is the number operator, and $k_{pq}^{\mu}$ and $j_{pq, \sigma \tau}^\mu$ denote matrix elements. The former matrix $k$ is complex and anti-Hermitian, and the latter matrix $j$ is real and symmetric. Assuming a Jordan-Wigner fermion-to-qubit mapping, each $e^{\pm K_\mu}$ orbital rotation can be realized in a quantum circuit with $O(N^2)$ Givens rotations on a device with linear connectivity. Implementing each $e^{i J_\mu}$ Jastrow operator also requires $O(N^2)$ two-qubit gates, but this assumes all-to-all connectivity or a fermionic swap network~\cite{motta_bridging_2023,jiang_quantum_2018}. The general structure of the UCJ ansatz is shown in Fig.~\ref{fig:ucj}.

The local UCJ (LUCJ) ansatz imposes sparsity on the $j_{pq, \sigma \tau}^\mu$ matrix, removing interactions that are unavailable or prohibitively expensive on a given quantum computer. For example, a common choice is $j_{p \alpha, q \beta} \mapsto j_{p \alpha, p \beta}$
for the opposite spin number-number terms, and
$
j_{p \sigma, q \sigma} \mapsto j_{p \sigma, (p + 1) \sigma}
$
for the same spin ($\sigma \in \{ \alpha, \beta \}$) number-number terms. The LUCJ circuit also contains $O(N^2)$ Givens rotations for each orbital rotation, and it contains $O(N)$ two-qubit gates for the sparsified Jastrow operation. Ref.~\cite{motta_bridging_2023} claims that, while the cost of simulating (L)UCJ circuits is polynomial on a quantum computer (from the above gate counts), it is super-polynomial on a classical computer.

\textit{Classical simulation algorithm} \label{sec:classical-simulation} ---
We work in second quantization and write the fermionic Hamiltonian as
\begin{equation} \label{eqn:hamiltonian-second-quantization}
    \H = E_0 + \sum_{pq, \sigma} h_{pq} a_{p \sigma}^\dagger a_{q \sigma} + \frac{1}{2} \sum_{pqrs, \sigma \tau} g_{pqrs} a_{p \sigma}^\dagger a_{q \tau}^\dagger a_{s \tau} a_{r \sigma} . 
\end{equation}
Our main result is the following:
\begin{theorem}(Polynomial time classical simulation of depth one UCJ circuits.)
    The energy
    \begin{equation}
        E_{\text{UCJ1}} := \langle \psi | \H | \psi \rangle_{\text{UCJ1}} 
    \end{equation}
    where $|\psi \rangle_{\text{UCJ1}}$ is an $L = 1$ UCJ ansatz~\eqref{eqn:ucj-ansatz} for Hamiltonian $\H$~\eqref{eqn:hamiltonian-second-quantization} with $N$ orbitals can be computed exactly on a classical computer in $O(N^7)$ time.
\end{theorem}
The algorithm uses Heisenberg evolution of the second-quantized Hamiltonian through the UCJ circuit, and proceeds in three steps. First (Lemma~\ref{lemma:backprop-through-orbital-rotation}) we backpropagate the one- and two-body tensors through the final orbital rotation in the circuit. Second (Lemma~\ref{lemma:backprop-through-jastrow}), we modify the backpropagated Hamiltonian to implement the Jastrow operation. We remark that this holds for any UCJ circuit, whether the Jastrow operation has been sparsified (LUCJ) or not. Last, we show that the final energy calculation can be executed efficiently through L\"{o}wdin's formula (Lemma~\ref{lemma:lowdin}) for computing matrix elements between non-orthogonal Slater determinants. 

\begin{lemma} \label{lemma:backprop-through-orbital-rotation}
    Propagating the second-quantized Hamiltonian $\H$~\eqref{eqn:hamiltonian-second-quantization} through the orbital rotations $e^{K_\alpha} e^{K_\beta}$ can be done exactly via updating the one-electron integrals
    \begin{equation} \label{eqn:one-electron-integrals-update-through-orbital-rotation}
        \tilde{h}_{p'q'} := \sum_{pq} u_{p'p} h_{pq} u_{q'q}^*
    \end{equation}
    in $O(N^3)$ time and by updating the two-electron integrals
    \begin{equation} \label{eqn:two-body-integrals-update-through-orbital-rotation}
        \tilde{g}_{p'q'r's'} := \sum_{pqrs} u_{p'p} u_{q'q} g_{pqrs} u_{r'r}^* u_{s's}^*
    \end{equation}
    in $O(N^5)$ time.
    Here,
    \begin{equation}
        u := e^{-k}
    \end{equation}
    is the unitary $N \times N$ matrix exponential of $k$, the final orbital rotation~\eqref{eqn:orbital-rotation-K_mu} in the UCJ circuit (see Fig.~\ref{fig:ucj}).
\end{lemma}
A proof of Lemma~\ref{lemma:backprop-through-orbital-rotation} is given in the Appendix. We remark that the number of terms in the backpropagated $\H$ does not change, only the elements of the one-electron and two-electron integrals $h$ and $g$, and this update is exact. As mentioned, this is in contrast to recent simulations of quantum circuits with Pauli or Majorana backpropagation in which the number of terms grows exponentially and truncation is used as a heuristic~\cite{refute_ibm_4}.

\begin{lemma} \label{lemma:backprop-through-jastrow}
    Propagating the second-quantized Hamiltonian $\H$~\eqref{eqn:hamiltonian-second-quantization} through a Jastrow operator can be done exactly in $O(N^5)$ time. This is done in two steps. The first step is updating the one-electron terms to 
    \begin{equation}
        \sum_{pq, \sigma} h_{pq} e^{i c_{p q \sigma}^{[1]}} a_{p \sigma}^\dagger a_{q \sigma} e^{i \boldsymbol{\phi}_{p q \sigma}^{[1]} \cdot \boldsymbol{\hat{n}}}
        %e^{-i J} \left( \sum_{pq, \sigma} h_{pq} a_{p\sigma}^\dagger a_{q \sigma} \right) e^{iJ} = \sum_{pq, \sigma} h_{pq} e^{i c_{p q \sigma}^{[1]}} a_{p \sigma}^\dagger a_{q \sigma} e^{i \boldsymbol{\phi}_{p q \sigma}^{[1]} \cdot \boldsymbol{\hat{n}}}
    \end{equation}
    where
    \begin{equation}
        c_{pq\sigma}^{[1]} := 2 j_{pq, \sigma \sigma} -j_{pp, \sigma\sigma} - j_{qq, \sigma \sigma} 
    \end{equation}
    and
    \begin{equation}
        \boldsymbol{\phi}_{p q \sigma}^{[1]} \cdot \boldsymbol{\hat{n}}  := \sum_{r\lambda} \phi_{r \lambda }^{(p q \sigma)} \hat{n}_{r \lambda} 
    \end{equation}
    with
    \begin{equation}
        \phi_{r \lambda }^{(p q \sigma)} := 2 \left(  - j_{pr, \sigma \lambda} + j_{qr, \sigma \lambda} \right) .
    \end{equation}
    This update preserves the $O(N^2)$ term count and can be done $O(N^3)$ time, since per term there is one $O(1)$ scalar phase $c_{pq \sigma}^{[1]}$ and one $O(N)$ vector phase $\phi_{r \lambda }^{(p q \sigma)}$.  The second step is updating the two-electron terms to
    \begin{equation} \label{eqn:two-body-electron-update-after-backpropagating-through-Jastrow-operator}
        \frac{1}{2} \sum_{pqrs, \sigma \tau} e^{i c_{pqrs, \sigma \tau}^{[2]}} g_{pqrs} a_{p \sigma}^\dagger a_{q \tau}^\dagger a_{s \tau} a_{r \sigma} e^{i \boldsymbol{\phi}_{p q r s, \sigma \tau}^{[2]} \cdot \boldsymbol{\hat{n}}}
    \end{equation}
    %
    % \begin{widetext}
    %     \begin{equation} \label{eqn:two-body-electron-update-after-backpropagating-through-Jastrow-operator}
    %         e^{-i J} \left( \frac{1}{2} \sum_{pqrs, \sigma \tau} g_{pqrs} a_{p \sigma}^\dagger a_{q \tau}^\dagger a_{s \tau} a_{r \sigma} \right) e^{iJ} =
    %         \frac{1}{2} \sum_{pqrs, \sigma \tau} e^{i c_{pqrs, \sigma \tau}^{[2]}} g_{pqrs} a_{p \sigma}^\dagger a_{q \tau}^\dagger a_{s \tau} a_{r \sigma} e^{i \boldsymbol{\phi}_{p q r s, \sigma \tau}^{[2]} \cdot \boldsymbol{\hat{n}}}
    %     \end{equation}
    % \end{widetext}
    %
    where
    \begin{equation}
    \footnotesize
    \begin{split}
        c_{pqrs, \sigma \tau}^{[2]} &:= 2 \left( - j_{pq,\sigma\tau} + j_{pr,\sigma\sigma} + j_{ps,\sigma\tau}
            + j_{qr,\tau\sigma} + j_{qs,\tau\tau} - j_{rs,\sigma\tau} \right) \\
            &\qquad - j_{pp,\sigma\sigma} - j_{qq,\tau\tau} - j_{rr,\sigma\sigma} - j_{ss,\tau\tau}
    \end{split}
\end{equation}
    and
    \begin{equation}
        \boldsymbol{\phi}_{p q r s, \sigma \tau}^{[2]} \cdot \boldsymbol{\hat{n}}  := \sum_{u \mu} \phi_{u \mu }^{(p q r s, \sigma \tau)} \hat{n}_{u \mu} 
    \end{equation}
    with
    \begin{equation}
        \phi_{u \mu }^{(p q r s, \sigma \tau)} := 2 \left( -\,j_{pu, \sigma \mu} - j_{qu, \tau \mu} + j_{ru, \sigma \mu} + j_{su, \tau \mu} \right).
    \end{equation}
    This update preserves the $O(N^4)$ term count and can be done in $O(N^5)$ time, since per term there is one $O(1)$ scalar phase $c_{pqrs, \sigma \tau}^{[2]}$ and one $O(N)$ vector phase $\phi_{u \mu }^{(p q r s, \sigma \tau)}$.
\end{lemma}
A proof of Lemma~\ref{lemma:backprop-through-jastrow} is given in the Appendix. We remark again that the number of terms does not change in the backpropagated Hamiltonian, but rather each term accumulates a scalar and vector phase.

The final step of our algorithm is computing the energy
\begin{equation}
    E_{\text{UCJ1}} := \langle \psi | \H | \psi \rangle_{\text{UCJ1}} = \langle Q | \tilde{\H} | Q \rangle 
\end{equation}
where $|Q\rangle := e^{-K} | \phi_0 \rangle $ is a Slater determinant and $\tilde{\H}$ is the backpropagated Hamiltonian after steps (Lemmas)~\ref{lemma:backprop-through-orbital-rotation} and~\ref{lemma:backprop-through-jastrow}. Because of the accumulation of the vector phase in Lemma~\ref{lemma:backprop-through-jastrow}, each term in the one-body and two-body part is a matrix element between non-orthogonal Slater determinants, which is handled by L\"{o}wdin's formula~\cite{Lowdin_1955}.

\begin{lemma}{(L\"{o}wdin~\cite{Lowdin_1955}.)} \label{lemma:lowdin}
    Matrix elements between non-orthogonal Slater determinants can be evaluated via
    \begin{equation}
        \langle Q | a_p^\dagger a_q e^{i \boldsymbol{\phi} \cdot \boldsymbol{\hat{n}}} | Q \rangle = \det (S) \,  \underbrace{\left( D_{\boldsymbol{\phi}} Q S^{-1} Q^\dagger \right)_{qp}}_{\rho_{qp}}
    \end{equation}
    where
    \begin{equation}
        D_{\boldsymbol{\phi}} := \text{diag}\left( e^{i\phi_1}, ..., e^{i \phi_N} \right)
    \end{equation}
    is an $N \times N$ diagonal matrix,
    \begin{equation} \label{eqn:s-matrix}
        S := Q^\dagger D_{\boldsymbol{\phi}} Q
    \end{equation}
    is an $N_{\text{occ}} \times N_{\text{occ}}$ matrix, and $|Q\rangle$ is a Slater determinant of $N_{\rm occ}$ orthonormal orbitals, the columns of the $N\times N_{\text{occ}}$  matrix Q. Similarly, for the two-electron terms,
    \begin{equation} \label{eqn:two-body-energy}
        \langle Q | a_p^\dagger a_q^\dagger a_s a_r e^{i \boldsymbol{\phi} \cdot \boldsymbol{\hat{n}}} | Q \rangle = \det (S) \, \left( \rho_{rp} \rho_{sq} - \rho_{sp} \rho_{rq} \right) .
    \end{equation}
\end{lemma}
Computationally, we form $S$ in $O(N^3)$ time and then compute its LU factorization in $O(N^3)$ time This allows the evaluation of $S^{-1} Q^\dagger$ in $O(N^3)$ time, and also provides $\det S$ in $O(N)$ time as the product of diagonal entries. For the one body terms, matrix elements $\rho_{qp}$ can be evaluated in $O(N^2)$ time via a triangular solve and dot product once the LU factorization of $S$ is performed ($O(N^3)$ time). Since there are $O(N^2)$ one body terms, we require $O(N^5)$ time to compute the energy of all one body terms. For the two body terms, the analysis is similar, and with $O(N^4)$ total two body terms we require $O(N^7)$ time. 

Although we never observe the matrix $S$~\eqref{eqn:s-matrix} to be singular in practice, in principle this is possible. A simple example illustrating this is $Q = [1 \ 1]^T / \sqrt{2}$ and $D_\phi = \text{diag}(1, -1)$. In the case that $S$ is singular, we can replace L\"{o}wdin's formula with the generalized non-orthogonal Wick's theorem, expressed in second quantization by Burton~\cite{Burton_2021} building on the generalized Slater–Condon rules. In this case, we compute a singular value decomposition (SVD) of $S$ and remove orbitals from $Q$ and phases from $D_\phi$ that correspond to zero singular values. Since the SVD costs $O(N^3)$ time, this has the same asymptotic runtime of $O(N^7)$. Because L\"{o}wdin's formula is somewhat simpler, we state the final step for energy computation in Lemma~\ref{lemma:lowdin} using this, but for the singular case we can use this procedure which returns the same energy and does not change the algorithm's runtime. In our numerical results below, we have tested both approaches and achieve the same energy results in all examples.

\begin{table}
    \centering
    \begin{tabular}{|c|c|} \hline
        \textbf{Algorithm}           & \textbf{Energy (Ha)} \\ \hline 
        Truncated Hex LUCJ-1.5\footnote{Identical to the circuit run in~\cite{Robledo2025}.}        & -323.051                    \\ \hline 
        % \textcolor{red!70!black}{Truncated Hex LUCJ-1.5 (Raw)} & -324.596              \\ \hline
        % \textcolor{red!70!black}{Truncated Hex LUCJ-1.5 + QED} &                -324.920 \\ \hline
        Hartree-Fock                 & -326.547                \\ \hline 
        Truncated Hex LUCJ-1             & -326.553       \\ \hline
        Hex LUCJ-1                       & -326.553 \\ \hline
        Square LUCJ-1               & -326.556 \\ \hline 
        UCJ-1                        & -326.559              \\ 
        \hline
        \textcolor{red!70!black}{Truncated Hex LUCJ-1.5 + SQD}~\cite{Robledo2025} & -326.645              \\ \hline
        Optimized UCJ-1  &  -326.796    \\ \hline 
        CCSD\footnote{Unconverged.}  & -326.828     \\ \hline 
        DMRG~\cite{Robledo2025}  & -327.239     \\ \hline 
        
    \end{tabular}
    \caption{Ground state energy calculation of the $n = 72$ qubit iron sulfur cluster computed by various algorithms, ordered from least negative (top row) to most negative (bottom row). The algorithm in \textcolor{red!70!black}{red} uses quantum hardware samples and SQD post-processing~\cite{Robledo2025}. By optimizing the parameters of a one layer UCJ circuit, we obtain a lower energy than experimentally achieved in~\cite{Robledo2025}.}
    \label{tab:iron-sulfur-results}
\end{table}

\begin{figure}[htbp]
    \centering
    \includegraphics[width=\linewidth]{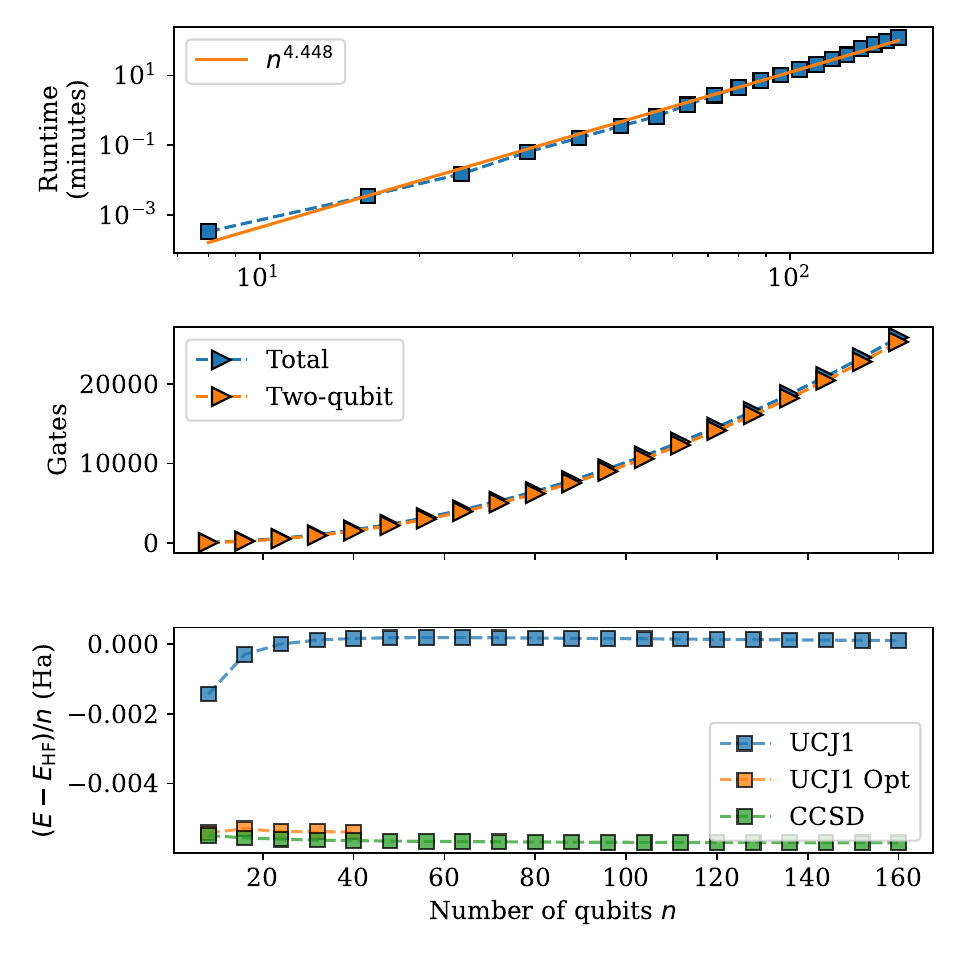}
    \caption{Algorithm runtime, circuit statistics, and energies for simulating the single-layer UCJ (UCJ1) ansatz for Hydrogen chains up to $80$ atoms ($160$ qubits). For visualization, we subtract the HF energy and normalize by the number of qubits $n$. As can be seen, the UCJ1 ansatz, initialized from CCSD parameters, closely follows Hartree-Fock energies. The maximum difference (in magnitude) between HF and UCJ1 is 16 mHa, and the average difference (in magnitude) is 12 mHa with a standard deviation of 5 mHa. We also optimize the parameters of the UCJ1 ansatz up to $40$ qubits and obtain comparable energies to CCSD.}
    \label{fig:hchains}
\end{figure}

\textit{Results} --- % \label{sec:simulation-of-experiments}
We first apply our algorithm to computing the ground state energy of the $n = 72$ qubit iron sulfur cluster experiment from~\cite{Robledo2025}. Energy results are shown in Table~\ref{tab:iron-sulfur-results}. On a 2025 Macbook Pro M5 with 16GB memory, we are able to classically simulate the same circuit, and variants, in less than a minute. Note that our algorithm applies to the extra ``half layer'' added to the ansatz in the experiment which only affects the final orbital rotations and does not change our algorithm. Additionally, we are able to simulate the full UCJ circuit with all Jastrow interactions.  The best energy obtained from quantum hardware data is $-326.645$ and comes from SQD run for 1.44 hours on 6400 nodes of Fugaku (Fig. S12 of~\cite{Robledo2025}). The best energy we obtain is $-326.796$ Ha and comes optimizing the single-layer UCJ circuit with L-BFGS-B~\cite{Liu_Nocedal_1989} for 2367 iterations. On an NVIDIA GH200 Grace Hopper node, energy evaluation takes around seven seconds, and the total runtime for the optimization was three hours.

% -Diagonalization over valid bit strings:-324.919647 Ha
% - 1647/3,163,742 - valid bit strings count
% - Diagonalization over all bit strings by taking 50,000 samples (uniform sampling)and projecting matrix over it and find lowest e.v, repeating the same 10 times and avg value is:-324.595799 Ha

Next, we apply our algorithm to computing the ground state energy of Hydrogen chains, a common molecular system for studying scaling, on up to 160 qubits --- the size of the current largest quantum computers available from IBM. Figure~\ref{fig:hchains} shows the runtime and energies, as well as the number of gates in the circuit. The best fit to the runtime in this case is $N^{4.448}$, which is smaller than the worst case $O(N^7)$ due to memoization (caching) of computed matrix elements. The energies obtained from the single-layer UCJ ansatz initialized from the CCSD parameters closely match the Hartree-Fock energy. The maximum difference (in magnitude) between HF and UCJ1 is 16 mHa, and the average difference (in magnitude) is 11 mHa with a standard deviation of 5 mHa. We also optimize circuit parameters using 1000 iterations of the L-BFGS-B optimizer up to $n = 40$ qubits, and obtain energies within a few mHa of CCSD, shown in Fig.~\ref{fig:hchains}.

\textit{Conclusion} ---
In this Letter, we introduced an exact, polynomial time classical algorithm for calculating the energy of any single-layer unitary cluster Jastrow (UCJ) circuit, which many recent large-scale quantum computations in chemistry have executed. Notably, we showed our algorithm can simulate the largest experiment of~\cite{Robledo2025} in less than a minute on a laptop. It is worth stating that our algorithm computes the energy instead of sampling bitstrings from the quantum circuit, a distinction which is often referred to in this context as weak simulation (computing energy or more generally observables) vs. strong simulation (sampling bitstrings). This distinction is important for interpreting our work in light of the previously mentioned hardness analysis of sampling from UCJ circuits~\cite{Hafid_Iwakiri_Tsubouchi_Yoshioka_Kohda_2025}. In strong simulation, can use sampled bitstrings to perform other tasks, and this is done in the experiment of~\cite{Robledo2025}; namely, through the sample-based quantum diagonalization (SQD) post-processing. This post-processing helps to remove errors from the quantum computer and improve the energy estimate. Our algorithm does not have the capability of sampling bitstrings, and so this post-processing is not possible. Nonetheless, we showed that circuit optimization, enabled by fast energy calculation, lowered the energy further than the SQD post-processing.

Our work demonstrates that $L \ge 2$ layers of the (L)UCJ ansatz is necessary to achieve any beyond-classical computation. Our algorithm does not hold for $L \ge 2$  layers because orbital rotations between Jastrow operations perturb the diagonal structure that enables L\"{o}wdin's formula for calculating the energy. It is unlikely that there exists any efficient classical algorithm for this case, as layers of alternating operators are well-known to be universal~\cite{Lloyd_1995}, and in particular Givens rotations and CZ gates are known to be universal for quantum chemistry~\cite{Leimkuhler_Whaley_2026}. Although the number of operations required for $L \ge 2$ layers on current quantum computers is challenging (which is a main reason why many prior experiments have used a single layer) our work shows that this will be necessary, but perhaps not sufficient, to achieve quantum advantage in chemistry. 
% In future work, it will be interesting to see if approximate Gaussian rank frameworks can classically simulate such UCJ circuits, albeit by paying an exponential cost in the number of non-Gaussian unitaries.

\textit{Code and data availability} ---
An implementation of our algorithm and script for simulating the iron sulfur cluster circuit from~\cite{Robledo2025} is available at~\cite{simulate_ucj1_github}.

\textit{Acknowledgments} ---
This material is based upon work supported by the National Science Foundation under Grant No. 2349002.
SD would like to acknowledge the financial support from Anusandhan National Research Foundation (ANRF) under the MATRICS scheme (Grant No.
[ANRF/ARGM/2025/002511/TS)]) and National Quantum
Mission under Quantum Algorithms Technical Group (TPN
No.: 136428). The work of SD has also been supported by
the International Center for Theoretical Sciences (ICTS), Bengaluru through its Associateship Programme. 
This collaborative study is supported by the Ministry of Education, Government of India, through the SPARC program (Project Code: SPARC/2025-2026/P4086). We acknowledge the use of generative AI (Claude Fable 5) for suggesting L\"{o}wdin's formula for computing the energy (Lemma~\ref{lemma:lowdin}). This work was supported in part through computational resources and services provided by the Institute for Cyber-Enabled Research at Michigan State University.

\bibliographystyle{apsrev4-1}
\bibliography{refs}

\appendix

\section{Proof of Lemma~\ref{lemma:backprop-through-orbital-rotation}} \label{sec:proof-lemma-1}

\begin{proof}
    
    We want to compute
    \begin{equation}
        \H \mapsto e^{-K_\alpha} e^{-K_\beta} \H e^{K_\alpha} e^{K_\beta} .
    \end{equation}
    For simplicity let us consider one spin sector and, for brevity, drop the spin variable so that we have
    \begin{equation} \label{eqn:orbital-rotation-conjugation-to-evaluate}
        \H \mapsto e^{-K} \H e^{K} % = E_0 + \sum_{pq} h_{pq} e^{-K} a_p^\dagger a_q e^K + \frac{1}{2} \sum_{pqrs} g_{pqrs} e^{-K} a_p^\dagger a_q^\dagger a_s a_r e^{K} .
    \end{equation}
    with
    \begin{equation}
        K = \sum_{pq} k_{pq} a_p^\dagger a_q .
    \end{equation}
    Note that this is the same as~\eqref{eqn:orbital-rotation-K_mu} with the spin index $\sigma$ dropped and the final layer index $\mu = L$ suppressed for brevity.
    By the Campbell identity/Hadamard formula,
    \begin{equation} \label{eqn:campbell-identity-hadamard-formula}
        e^X Y e^{-X} = \sum_{n = 0}^{\infty} \frac{1}{n!} [X, Y]_n
    \end{equation}
    where
    \begin{equation}
    [X,Y]_n := \underbrace{[X,\cdots[X,[X}_{n \text{ times}},Y]]\cdots], \quad [X,Y]_0 := Y ,
    \end{equation}
    we can evaluate~\eqref{eqn:orbital-rotation-conjugation-to-evaluate} by evaluating nested commutators. Consider first the application of the Campbell identity to a single creation operator in $K$, i.e.
    \begin{equation} \label{eqn:conjugation-of-creation-operator-by-exp-K}
        e^{-K} a_p^\dagger e^{K} = \sum_{n = 0}^{\infty} \frac{1}{n!} [-K, a_p^\dagger]_n
        = \sum_{n = 0}^{\infty} \frac{(-1)^n}{n!} [K, a_p^\dagger]_n.
    \end{equation}
    To evaluate $[K, a_p^\dagger]_n$, we have for the base case $n = 1$ that
    \begin{equation} \label{eqn:commutator-of-K-with-creation-base-case}
        [K, a_p^\dagger] = \sum_{rs} k_{rs} [a_r^\dagger a_s, a_p^\dagger] .
    \end{equation}
    Using the commutator-anticommutator identity 
    \begin{equation} \label{eqn:commutator-anticommutator-identity}
        [AB, C] = A\{B, C\} - \{A, C\}B ,
    \end{equation}
    we have
    \begin{equation} \label{eqn:inner-commutator-derivation}
        [a_r^\dagger a_s, a_p^\dagger] = a_r^\dagger \{a_s, a_p^\dagger\} - \{a_r^\dagger, a_p^\dagger\} a_s = \delta_{sp} a_r^\dagger
    \end{equation}
    where in the second equality we have used the canonical commutation relations
    \begin{equation}
        \{a_s, a_p^\dagger\} = \delta_{sp}, \qquad \{a_s, a_p\} = \{a_s^\dagger, a_p^\dagger\} = 0 .
    \end{equation}
    Inserting~\eqref{eqn:inner-commutator-derivation} into~\eqref{eqn:commutator-of-K-with-creation-base-case}, we see that
    \begin{equation}
        [K, a_p^\dagger] = \sum_{r} k_{rp} a_r^\dagger = \sum_r (k^T)_{pr} a_r^\dagger = [k^T \boldsymbol{a}^\dagger]_p ,
    \end{equation}
    where in the last step we have written the expression as the $p$th element of the matrix-vector product $k^T \boldsymbol{a}^\dagger$. Now,
    by repeated application, it follows that
    \begin{equation}
        [K, a_p^\dagger]_n = \sum_r [k^n]_{rp} a_r^\dagger ,
    \end{equation}
    and so from~\eqref{eqn:conjugation-of-creation-operator-by-exp-K}, we have that
    \begin{equation}
        e^{-K} a_p^\dagger e^{K} = \sum_{n = 0}^\infty \frac{(-1)^{n}}{n!} \sum_r k_{rp}^n a_r^\dagger = \sum_r \left( \sum_{n = 0}^\infty \frac{(-1)^{n}}{n!} [k^n]_{rp} \right) a_r^\dagger %\left[ e^{-k^T} \boldsymbol{a}^\dagger \right]_p .
    \end{equation}
    We define
    \begin{equation}
        u_{rp} := \sum_{n = 0}^\infty \frac{(-1)^{n}}{n!} [k^n]_{rp} ,
    \end{equation}
    or as a matrix
    \begin{equation}
        u := e^{-k},
    \end{equation}
    so that we can write
    \begin{equation} \label{eqn:conjugation-of-creation-operator-by-orbital-rotation}
        e^{-K} a_p^\dagger e^{K} = \sum_r u_{rp} a_r^\dagger .
    \end{equation}
    By taking the hermitian conjugate of this, we see that
    \begin{equation} \label{eqn:conjugation-of-annihilation-operator-by-orbital-rotation}
        e^{-K} a_q e^{K} = \sum_s u_{sq}^* a_s .
    \end{equation}
    Thus, the evolution of the one-electron term in~\eqref{eqn:orbital-rotation-conjugation-to-evaluate} becomes
    \begin{align}
        e^{-K} \left( \sum_{pq} h_{pq} a_p^\dagger a_q \right) e^K &= \sum_{pq} h_{pq} e^{-K} a_p^\dagger a_q e^{K} \\
        &= \sum_{pq} h_{pq} \left( e^{-K} a_p^\dagger e^K \right) \left( e^{-K} a_q e^{K} \right) \\
        &= \sum_{pq} h_{pq} \left( \sum_r u_{rp} a_r^\dagger \right) \left( \sum_s u_{sq}^* a_s \right) \\
        &= \sum_{rs} \left( \sum_{pq} u_{rp} h_{pq} u_{sq}^* \right) a_r^\dagger a_s \\
        &= \sum_{rs} \tilde{h}_{rs} a_r^\dagger a_s .
    \end{align}
    In other words, to apply the orbital rotation to the Hamiltonian, we simply need to update one-body integrals $h_{pq}$ via
    \begin{equation}
        \tilde{h}_{rs} = \sum_{pq} u_{rp} h_{pq} u_{sq}^*
    \end{equation}
    which is exactly~\eqref{eqn:one-electron-integrals-update-through-orbital-rotation}, 
    or as a matrix equation
    \begin{equation}
        \tilde{h} = u h u^\dagger .
    \end{equation}
    Thus, the one-body integrals can be updated exactly in time $O(N^3)$. 

    The two-body integrals are handled similarly. Writing
    \begin{align}
        e^{-K} \left( \frac{1}{2} \sum_{pqrs} g_{pqrs}\, a_p^\dagger a_q^\dagger a_s a_r \right) e^{K}
        &= \frac{1}{2} \sum_{pqrs} g_{pqrs}\, \tilde a_p^\dagger\, \tilde a_q^\dagger\, \tilde a_s\, \tilde a_r
    \end{align}
    where the tilde denotes conjugation by $e^{-K}$ --- e.g., $\tilde a_p^\dagger \equiv e^{-K} a_p^\dagger e^K$. Using~\eqref{eqn:conjugation-of-creation-operator-by-orbital-rotation} and~\eqref{eqn:conjugation-of-annihilation-operator-by-orbital-rotation}, we arrive at the updated two-body integrals~\eqref{eqn:two-body-integrals-update-through-orbital-rotation}, which can be computed in $O(N^5)$ time by contracting one index at a time.
    
\end{proof}

\section{Proof of Lemma~\ref{lemma:backprop-through-jastrow}} \label{sec:proof-lemma-2}

\begin{proof}
    To evaluate
    \begin{equation}
        \mathcal{H} \mapsto e^{-i J} \mathcal{H} e^{iJ}, 
    \end{equation}
    we again use the Campbell identity/Hadamard formula~\eqref{eqn:campbell-identity-hadamard-formula}. Starting with the conjugation of a creation operator, we have
    \begin{equation} \label{eqn:temp-number-operator-creation-operator-commutator}
        [\hat{n}_{r\lambda}, a_{p \sigma}^\dagger] = [a_{r\lambda}^\dagger a_{r \lambda}, a_{p \sigma}^\dagger] = \delta_{rp} \delta_{\lambda \sigma} a_{p \sigma}^\dagger 
    \end{equation}
    where the last equality follows from the commutator-anticommutator identity~\eqref{eqn:commutator-anticommutator-identity} and canonical commutation relations. It follows that
    \begin{equation}
        [J, a_{p \sigma}^\dagger] = \sum_{rs, \lambda \tau} j_{rs, \lambda \tau} [\hat{n}_{r \lambda} \hat{n}_{s \tau}, a_{p \sigma}^\dagger ] .
    \end{equation}
    We evaluate the inner commutator similarly by expanding and using~\eqref{eqn:temp-number-operator-creation-operator-commutator} to obtain
    \begin{equation}
        \scriptsize
        [\hat{n}_{r\lambda} \hat{n}_{s\tau}, a_{p\sigma}^\dagger] = \delta_{sp}\delta_{\tau\sigma}\, \hat{n}_{r\lambda} a_{p\sigma}^\dagger + \delta_{rp}\delta_{\lambda\sigma}\, \hat{n}_{s\tau} a_{p\sigma}^\dagger - \delta_{rp}\delta_{\lambda\sigma}\delta_{sp}\delta_{\tau\sigma}\, a_{p\sigma}^\dagger .
    \end{equation}
    Inserting back into the summation and using the fact that $j$ is symmetric, we have
    \begin{equation}
        [J, a_{p \sigma}^\dagger] = \hat{\Lambda}_{p \sigma} a_{p \sigma}^\dagger
    \end{equation}
    where $\hat{\Lambda}_{p \sigma}$ is the operator defined via
    \begin{equation} \label{eqn:Lambda}
        \hat{\Lambda}_{p \sigma} := 2 \sum_{r \lambda} j_{pr, \sigma \lambda} \hat{n}_{r \lambda} - j_{pp, \sigma \sigma} .
    \end{equation}
    It follows by induction that the nested commutator is
    \begin{equation}
        [J, a_{p \sigma}^\dagger]_n = \hat{\Lambda}_{p \sigma}^n a_{p \sigma}^\dagger
    \end{equation}
    so that
    \begin{align}
        e^{-iJ} a_{p \sigma}^\dagger e^{iJ} &= \sum_{n = 0}^{\infty} \frac{(-i)^n}{n!} [J, a_{p \sigma}^\dagger]_n \\
        &= \sum_{n = 0}^{\infty} \frac{(-i)^n}{n!} \hat{\Lambda}_{p \sigma}^n a_{p \sigma}^\dagger \\
        &= e^{-i \hat{\Lambda}_{p \sigma}}  a_{p \sigma}^\dagger .
    \end{align}
    Now, taking the hermitian conjugate, we immediately have
    \begin{equation}
        e^{-iJ} a_{q \tau} e^{iJ} =  a_{q \tau} e^{i \hat{\Lambda}_{q \tau}} .
    \end{equation}
    Thus, the evolution of the one-electron integrals in~\eqref{eqn:orbital-rotation-conjugation-to-evaluate} becomes
    \begin{align}
        e^{-iJ} \sum_{pq,\sigma} h_{pq} a_{p \sigma}^\dagger a_{q \sigma} e^{iJ} &= \sum_{pq, \sigma} h_{pq} \left( e^{-iJ} a_{p \sigma}^\dagger e^{iJ} \right) \left( e^{-iJ} a_{q \sigma} e^{iJ} \right) \\
        &= \sum_{pq, \sigma} h_{pq} e^{-i \hat{\Lambda}_{p \sigma}}  a_{p \sigma}^\dagger a_{q \sigma} e^{i \hat{\Lambda}_{q \sigma}} \label{eqn:temp-one-body-integrals-conjugation-by-jastrow}
    \end{align}
    From~\eqref{eqn:Lambda} and the commutator~\eqref{eqn:temp-number-operator-creation-operator-commutator}, we have that
    \begin{equation}
        \hat{\Lambda}_{p\sigma} a_{s \mu}^\dagger = a_{s \mu}^\dagger (\hat{\Lambda}_{p\sigma} + 2 j_{ps, \sigma \mu})
    \end{equation}
    and thus
    \begin{equation}
        e^{-i\hat{\Lambda}_{p\sigma}} a_{s \mu}^\dagger = a_{s \mu}^\dagger e^{-i(\hat{\Lambda}_{p\sigma} + 2 j_{ps, \sigma \mu})} .
    \end{equation}
    By similar reasoning,
    \begin{equation}
        e^{-i\hat{\Lambda}_{p\sigma}} a_{s \mu} = a_{s \mu} e^{-i(\hat{\Lambda}_{p\sigma} - 2 j_{ps, \sigma \mu})} .
    \end{equation}
    We apply these to push the exponential of $e^{-i\hat{\Lambda}_{p\sigma}}$ to the right of the creation and annihilation operators in~\eqref{eqn:temp-one-body-integrals-conjugation-by-jastrow}. When we combine with the $e^{i\hat{\Lambda}_{q\sigma}}$ on the right, we have
    \begin{equation}
        e^{-i \hat{\Lambda}_{p \sigma}}  a_{p \sigma}^\dagger a_{q \sigma} e^{i \hat{\Lambda}_{q \sigma}} = a_{p \sigma}^\dagger a_{q \sigma} e^{i(\hat{\Lambda}_{q\sigma} - \hat{\Lambda}_{p\sigma} - 2 j_{pp, \sigma \sigma} + 2 j_{pq, \sigma \sigma})} .
    \end{equation}
    Using the definition of $\hat{\Lambda}$~\eqref{eqn:Lambda}, we can write
    \begin{equation}
        e^{-i \hat{\Lambda}_{p \sigma}}  a_{p \sigma}^\dagger a_{q \sigma} e^{i \hat{\Lambda}_{q \sigma}} = e^{i c_{p q \sigma}^{[1]}} a_{p \sigma}^\dagger a_{q \sigma} e^{i \boldsymbol{\phi}_{p q \sigma }^{[1]} \cdot \boldsymbol{\hat{n}}}
    \end{equation}
    where we define the scalar
    \begin{equation}
        c_{pq\sigma}^{[1]} := 2 j_{pq, \sigma \sigma} -j_{pp, \sigma\sigma} - j_{qq, \sigma \sigma} 
    \end{equation}
    and the operator
    \begin{equation}
        \boldsymbol{\phi}_{p q \sigma}^{[1]} \cdot \boldsymbol{\hat{n}} := \sum_{r\lambda} \phi_{r \lambda }^{(p q \sigma)} \hat{n}_{r \lambda}
    \end{equation}
    with
    \begin{equation}
        \phi_{r \lambda }^{(p q \sigma)} := 2 \left( j_{qr, \sigma \lambda} - j_{pr, \sigma \lambda} \right)
    \end{equation}
    The update to the two-body part of the Hamiltonian~\eqref{eqn:two-body-electron-update-after-backpropagating-through-Jastrow-operator} can be derived similarly.
\end{proof}

\section{Hardness of other simulation algorithms for single-layer UCJ circuits}

\subsection{Matchgate simulation}

As mentioned in the main text, matchgate + CPHASE simulation has been implemented for LUCJ circuits~\cite{Hassman_ReardonSmith_Ravi_Chong_Sung_2025}. To our knowledge, the largest matchgate + CPHASE simulation of LUCJ circuits comes from this reference and is for $n = 52$ qubit circuit with $57$ CPHASE gates.
This simulation technique is best suited for LUCJ circuits which truncate CPHASE gates from the Jastrow operation. UCJ circuits would be significantly more challenging. For example, the full UCJ circuit for the $n = 72$ qubit iron sulfur cluster from~\cite{Robledo2025} contains 2556 CPHASE gates, which is well beyond the capabilities of this simulation technique. Generally, the $O(N)$ scaling and $O(N^2)$ scaling of the number of gates in the LUCJ and UCJ ansatzes, respectively, make matchgate + CPHASE simulation intractable.

\subsection{Heisenberg simulation (Observable propagation)}

To estimate an expectation value
\begin{equation} \label{eqn:expectation-value}
    \mathbb{E}_{\rho, U} [ O ] := \Tr[U \rho U^\dagger O]
\end{equation}
Heisenberg simulators write
\begin{equation}
    \mathbb{E}_{\rho, U} [ O ] = \Tr[ \rho \left( U^\dagger O U \right)]
\end{equation}
and apply the unitary $U$ to the observable $O$. We consider $U$ to be a quantum circuit composed of $N$ gates $U = U_N U_{N - 1} \cdots U_1$, and so the task of Heisenberg simulation is to \textit{propagate} the observable
\begin{equation}
    O \mapsto U_1 ^\dagger O U_1 \mapsto U_N^\dagger \cdots U_1^\dagger O U_1 \cdots U_N .
\end{equation}

Pauli propagation is Heisenberg evolution (observable propagation) with the observable $O$ expressed in the Pauli basis
\begin{equation}
    O = \sum_i c_i P_i 
\end{equation}
where $P_i$ are $n$-qubit Pauli operators of the form
\begin{equation}
    P = s X_1^{x_1} X_2^{x_2} \cdots X_n^{x_n} Z_1^{z_1} Z_2^{z_2} \cdots Z_n^{z_n}
\end{equation} 

Majorana propagation is Heisenberg evolution (observable propagation) with the observable $O$ expressed in the Majorana basis
\begin{equation}
    O = \sum_i c_i M_i
\end{equation}
where $M_i$ are $n$-qubit Majorana operators (often called Majorana monomials in this context) of the form
\begin{equation}
    M = s \gamma_1^{b_1} \gamma_2^{b_2} \cdots \gamma_{2n - 1}^{b_{2n - 1}} \gamma_{2n}^{b_{2n}} .
\end{equation}
Here, $s \in \mathbb{C}$ is a phase, the Majorana operators are defined by
\begin{align} \label{eqn:majorana-operators-odd}
    \gamma_{2j - 1} &:= \left( \prod_{k = 1}^{j - 1} Z_k \right) X_j \\
    \gamma_{2j} &:= \left( \prod_{k = 1}^{j - 1} Z_k \right) Y_j , \label{eqn:majorana-operators-even}
\end{align}
and $b_i \in \{0, 1\}$. 

We have implemented Pauli and Majorana propagation algorithms for the (L)UCJ circuits we study in this work. In general, we find that these algorithms require a truncation that is too large to obtain chemically relevant results (e.g., better than Hartree-Fock). For smaller truncation, the number of terms and runtime grows exponentially, presenting a significant challenge for these algorithms.

\begin{figure}[htbp]
    \centering
    \includegraphics[width=\linewidth]{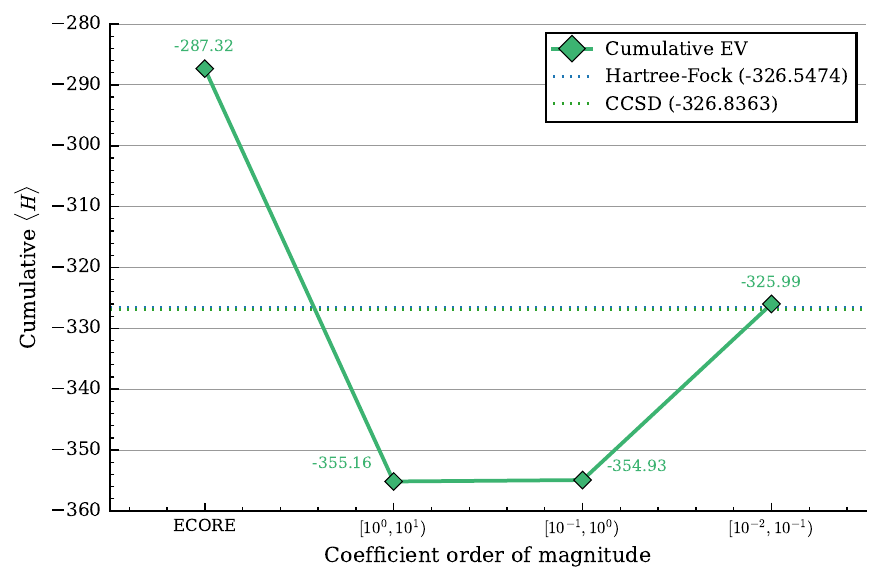}
    \caption{Cumulative energy of Majorana propagation for the $n = 72$ qubit iron sulfur cluster Hamiltonian of~\cite{Robledo2025} through a 1 layer LUCJ ansatz with heavy-hex connectivity. Here, energy is computed by Majorana backpropagation with a cutoff of $10^{-6}$ for terms partitioned by coefficient magnitude (see Fig.~\ref{fig:hamiltonian}). The final cumulative energy is not yet converged, and over $500$ mHa from the Hartree-Fock energy, demonstrating the challenge of this approach. More accurate simulations (larger cutoffs) are prohibitively expensive, as the number of terms grows exponentially with the number of gates in the circuit, and truncation is used as a heuristic.}
    \label{fig:majorana}
\end{figure}

\begin{figure}[htbp]
    \centering
    \includegraphics[width=\linewidth]{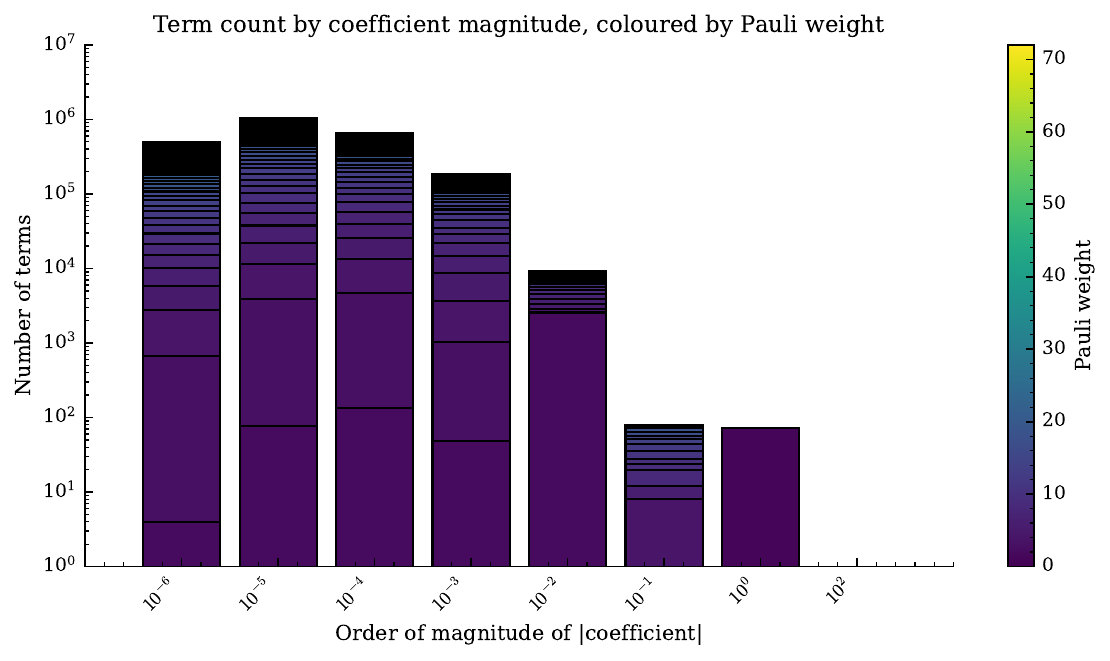}
    \caption{Partitioning of the 72 qubit iron sulfur cluster Hamiltonian from~\cite{Robledo2025}, showing the number of terms in the Hamiltonian vs coefficient magnitude, colored by Pauli weight.}
    \label{fig:hamiltonian}
\end{figure}

An example is shown in Fig.~\ref{fig:majorana}. Here, we compute the energy by mapping fermionic operators to Majorana operators and backpropagating these through the circuit. The Hamiltonian contains over 2.5M terms in the Majorana basis, which presents a significant practical overhead. To deal with this, we partition the Hamiltonian into sectors based on coefficient magnitude and/or term weight (Pauli weight). Figure~\ref{fig:hamiltonian} shows the Hamiltonian partitioned into sectors based on coefficient magnitude, and colored by term weight (Pauli weight). We see that many terms have low weight, and most terms have small coefficient magnitude. Recognizing this, we backpropagate Majorana terms with the largest coefficient magnitude from $10^{-2}$ to $10^{2}$, using a cutoff in the simulation of $10^{-6}$. That is, Majorana terms with coefficient magnitude smaller than $10^{-6}$ are discarded in the backpropagation. The cumulative energy vs coefficient magnitude order is shown in Fig.\ref{fig:majorana}. Here, we observe a large variance in cumulative energy vs. coefficient magnitude, and the plot is not yet converged. Computing more terms until convergence would take significantly more computational resources beyond our current capabilities, and even then there is no guarantee that the cutoff of $10^{-6}$ will lead to chemically accurate energies. After the $10^{-1}$ coefficients computed in Fig.~\ref{fig:majorana}, the energy is not yet converged and over $500$ mHa from the Hartree-Fock energy. We observe similar results if we use the Pauli basis instead of the Majorana basis.

\subsection{Tensor networks}

\begin{figure}
    \centering
    \includegraphics[width=\linewidth]{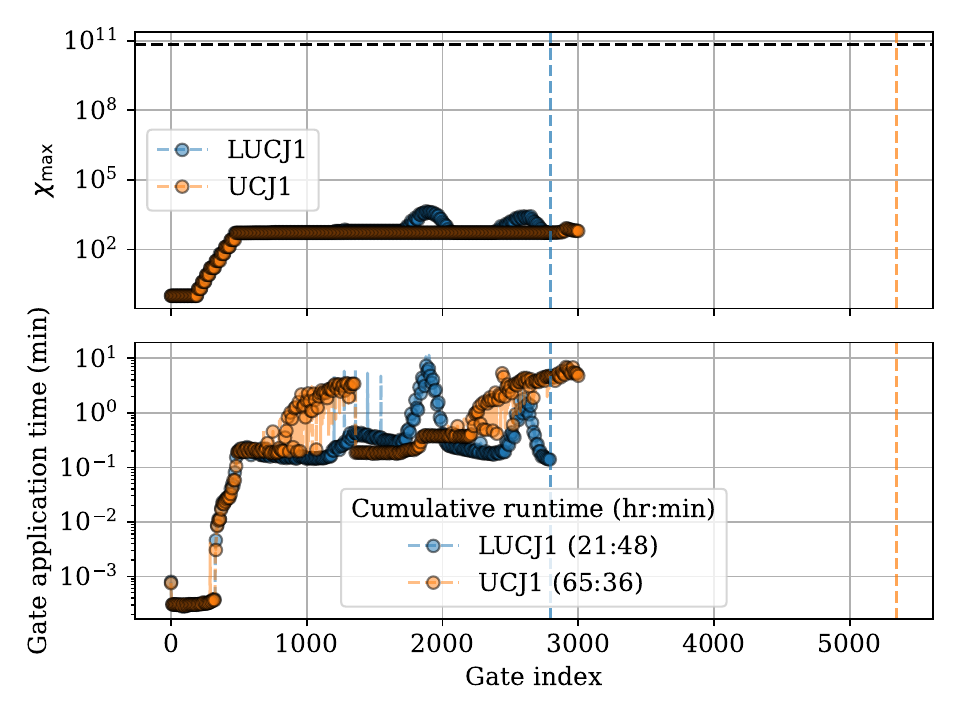}
    \caption{Analysis of MPS TEBD for the $n = 72$ qubit iron sulfur LUCJ1 and UCJ1 circuits. The top plot shows the maximum bond dimension and the bottom plot shows the gate application time vs gate index. Dashed vertical lines represent the final gate index for each circuit. The horizontal dashed black line in the top plot shows the maximum possible bond dimension of $2^{n / 2}$.}
    \label{fig:mps}
\end{figure}

We have also analyzed the cost of tensor network algorithms to simulate single-layer (L)UCJ circuits, in particular the well-known time-evolving block decimation (TEBD) algorithm for matrix product states (MPS). The structure of UCJ circuits (see Fig.~\ref{fig:ucj}) lends nicely to this simulation strategy, as the circuit is bipartite until the opposite spin Jastrow interactions. This means that two separate MPS evolutions can be carried out, and in practice we find that the bond dimension is relatively small for the CCSD initialized UCJ circuits we test. The challenging part for this simulation algorithm is the opposite spin Jastrow interactions. If most of these operations are removed as in the truncated LUCJ case, then MPS simulation can be effective, and we observe this in simulation. For full UCJ circuits with unrestricted opposite spin Jastrow interactions, the simulation becomes significantly harder and intractable for MPS TEBD.

To illustrate this, we perform MPS simulation of the single-layer LUCJ and UCJ circuits for the iron sulfur cluster experiment from~\cite{Robledo2025}. In Fig.~\ref{fig:mps}, we show the maximum bond dimension $\chi_\text{max}$ and the gate application time for each gate in the circuits. As can be seen, bond dimension quickly grows but then saturates during the initial orbital rotations. After this, the truncated LUCJ1 circuit entanglement grows during the Jastrow interactions, as expected. During this portion of the circuit, we see the entanglement reduce and then start to grow again. These bond dimensions are well below the maximum possible for an $n$ qubit circuit, but they are still large by practical considerations. Nonetheless, we are able to compute the final MPS for the LUCJ1 circuit in just under one day of runtime on a compute node with AMD EPYC 7H12 processors, 128 cores, and 412 GB available memory. From this, samples could be drawn from the MPS to compute the energy with SQD post-processing, in the same way as in~\cite{Robledo2025}.

For the UCJ1 ansatz, however, the situation is different. While the LUCJ1 circuit contains 75 CPHASE gates in the Jastrow operation, the unrestricted UCJ1 ansatz contains 2556 CPHASE gates in the Jastrow operation. Fully simulating this circuit will almost surely be out of reach for current classical computers. Indeed, after 65 hours on the same compute node described above, we were not able to get past the Jastrow operation. Interestingly, the entanglement has not grown appreciably during this portion of the circuit. Still, the non-local gates take significant time to implement due to the required chain of SWAP gates. Similar to matchgate + CPHASE algorithm, the distinction between local (LUCJ) and non-local significantly changes the computational hardness of simulating the circuit. Even if this circuit is simulable, different parameters in the circuit affecting entanglement growth and/or scaling to larger numbers of qubits will be challenging. In contrast, our algorithm working directly in the fermionic basis is able to handle any single-layer UCJ circuit, independent of locality, in polynomial time.

\end{document}